\documentclass[english,superscriptaddress,prl,twocolumn]{revtex4}
\usepackage{amsthm}
\usepackage[toc,page,title,titletoc,header]{appendix}
\usepackage[T1]{fontenc}
\usepackage{color}
\usepackage{tikz}
\usepackage{enumerate}
\usepackage{bm}
\usepackage{graphicx}
\usepackage[utf8]{inputenc}
\usepackage{epstopdf}
\usepackage{amsthm}

\usepackage{amsmath}
\usepackage{shapepar}
\usepackage{color}
\usepackage{comment}

\usepackage{amssymb, amsbsy, amsthm}
\makeatletter

\newtheorem{Theorem}{Theorem}
\newtheorem{Lemma}{Lemma}

\usepackage{babel}
\makeatother
\begin{document}
\title{Entanglement Detection Beyond Local Bound with Coarsely  Calibrated measurements}

\author{Liang-Liang Sun}\email{sun18@ustc.edu.cn}
\affiliation{Department of Modern Physics and National Laboratory for Physical Sciences at Microscale, University of Science and Technology of China, Hefei, Anhui 230026, China}

\author{Yong-Shun Song}
\affiliation{School of Information Engineering, Changzhou Vocational Institute of Industry Technology, Changzhou 213164, China}

\author{Sixia Yu }
\email{yusixia@ustc.edu.cn} 
\affiliation{Department of Modern Physics and National Laboratory for Physical Sciences at Microscale, University of Science and Technology of China, Hefei, Anhui 230026, China}
\affiliation{Hefei National Laboratory, University of Science and Technology of
China, Hefei 230088, China}

\date{\today{}}
\begin{abstract}


Bell's test, initially devised to distinguish quantum theory from local hidden variable models through {violations of local bounds}, is also a common tool for detecting entanglement. For this purpose, one can assume the quantum description of devices and use available information to strengthen the bound for separable states, which may go beyond the local bound, enabling more efficient entanglement detection.
Here we present a systematic approach for strengthening Bell inequalities for qubit systems, using Mermin-Klyshko-Bell inequalities as examples, by considering measurement devices that are coarsely calibrated only by their ability to generate nonlocal correlations without requiring precise quantum characterization.
In the case of bipartite and tripartite systems, we derive trade-offs between upper bounds for separable states and general states in terms of structure functions to optimize the entanglement detection. We then strengthen $n$-partite Bell inequalities for the detection of states exhibiting a diversity of entanglement structures such as genuine multipartite entanglement. For general Bell scenarios with some measurements characterized, we demonstrate that entanglement can also be detected  with some local correlations by exploiting the Navascu\'{e}s-Pironio-Ac\'{i}n hierarchy of tests.

\end{abstract}


\maketitle
\section{I. Introduction}
Quantum entanglement is a resource accounting for quantum advantage in many quantum information tasks, such as quantum computation and communication~\cite{Briegel2009} and metrology~\cite{PhysRevLett.67.661, Giovannetti2004}. One important issue in quantum information science is thus to detect quantum entanglement carried by a  quantum system, for which,  Bell's test~\cite{RevModPhys.86.419} is one of the most frequently used tools. Correlations violating Bell's inequalities can {firmly} verify the presence of entanglement.  {Furthermore}, in a multipartite system, the degree of Bell's violation can also indicate how many particles are genuinely entangled with each other, thus witnessing entanglement depth~\cite{PhysRevLett.86.4431, PhysRevLett.112.155304, PhysRevLett.90.080401}.

Standard Bell's inequalities, initially designed for the purpose of testing quantum theory against local hidden variables models~\cite{RevModPhys.86.419, GUHNE20091}, are formulated in terms of  experimental statistics and without involving device descriptions. For  entanglement detection, they  have inherent limitations. A single Bell inequality is typically violated by only a specific class of states. Furthermore, there are entangled states that do not demonstrate non-locality~\cite{PhysRevA.40.4277} and, thus, cannot be detected by any Bell test.
To overcome these difficulties, one can employ a quantum description of the systems and use available information to strengthen the bounds for separable states~\cite{PhysRevLett.98.140402, Uffink2008, PhysRevA.80.034302}. Assuming projective and orthogonal measurements, the bound of the Clauser-Horne-Shimony-Holt inequality (CHSH)~\cite{PhysRevLett.23.880} for separable states can be strengthened from the local bound of $2$ to $\sqrt{2}$, leading to more efficient detection and providing a quantitative statement on entanglement~\cite{PhysRevA.80.034302}. When only partial measurements are characterized, one approach based on the cross-norm~\cite{Rudolph2005} or realignment criterion~\cite{article} is exploited  in Ref.\cite{Moroder2012}. Another fruitful {approach} is generalizing the Bell scenario to a measurement-device-independent Bell scenario via introducing characterized sources on each local side ~\cite{PhysRevLett.118.150505}, which can, in principle, detect all entangled states. 
{While} these generalizations have significantly advanced the theory of quantum entanglement, a key practical open question remains: how to strengthen Bell's inequality using general measurements for multipartite qubit systems--the standard platform for quantum information science, for which,  previous considerations  mainly focused on the minimal Bell's scenario using  projective measurements~\cite{PhysRevLett.23.880, PhysRevA.80.034302}.

In this paper,  we tackle this issue by  {tightening Bell inequalities for multipartite qubit systems in two experiment settings.  In the first one, we consider measurements that are   coarsely calibrated with capability of generating nonlocal correlation,  which are  referred to as nonlocal-correlation-generating (NLCG) measurements}. In this setting,   we first consider bipartite and tripartite Bell's inequalities and derive trade-offs between upper bounds for local and general states. These trade-offs enable one to verify entanglement without requiring precise characterization of them except that they are NLCG, namely, capable of generating nonlocal correlations. We then consider $n$-partite Bell's inequalities due to Mermin-Klyshko (MK) \cite{PhysRevLett.65.1838, KLYSHKO1993399, PhysRevLett.88.230406}  and provide a systematic approach to derive analytical bounds for  states with various entanglement structures. {In the second experiment setting, we }  consider more general Bell's scenarios and adapt the Navascu\'{e}s-Pironio-Ac\'{i}n (NPA) hierarchy ~\cite{PhysRevLett.98.010401} for our purposes, with some measurements are characterized.


\section{II. strengthen  Bell tests for separable states. }

 {In this  section, we establish that, for  nonlocal correlation generating (NLCG) measurements, the bound of  the Bell inequality for a separable state $\varrho$  could be smaller than the classical bound. We focus on the $(n, 2, 2)$ Bell scenario, in which an $n$-partite state $\rho$ is distributed among $n$ spacelike-separated observers, on the received state the $m$-th  observer can randomly perform  one of the two local binary measurements $A_{s_m}$, with  $s_{m} \in \{0,  1\}$.  A Bell inequality can then be formulated in terms of expectation values of these jointly measured local observables.   }
$$\beta= \sum_{{\bf s}}\alpha_{\bf s} \overline {\otimes_{s_{m}}A_{s_m} }\leq L,$$
 where  ${\bf s}\equiv \{s_{1}, \cdots, s_{n}\}$ specifies the collection of  local  measurement settings and $\overline {\otimes_{s_{m}}A_{s_m} }  $ specifies the expectation value of jointly measured observables  $\{A_{s_{1}}, \cdots,  A_{s_{n}}\}$ and  $\alpha_{\bf  s}$ specifies the relevant coefficient   and $L$ is the classical bound of Bell's inequality.

{\bf  Definition}{\emph{ NLCG measurement}:   {A set of measurements $\Omega'$ is called NLCG measurements if,  for some state $\rho$, the expectation value of the relevant  Bell operator   $\hat{\beta}(\Omega') = \sum_{\bf s} \alpha_{\bf s} \otimes_{s_{m}} \hat{A}_{s_m}$   can be  larger than the classical bound   $\beta_{\rho}(\Omega') > L$ where    $\beta_{\rho}(\Omega') \equiv {\rm Tr}[\rho \hat{\beta}(\Omega')]$}.

 {To tighten the bound of Bell's inequality for separable states under general NLCG  measurements denoted  by  $\Omega'$, we proceed in two steps. First, considering  the set of rank one and  projective measurements $\Omega$, we derive the upper bounds of  Bell inequality for an entangled state and a general separable state $\varrho$, denoted $\mathrm{U}(\Omega)$ and $\mathrm{U}_{\varrho}(\Omega)$ respectively, and establish the trade-off  $\mathrm{U}_{\varrho}(\Omega) \leq f\left(\mathrm{U}(\Omega)\right) \leq L$, where $f(\cdot)$ is the structure function  defined as 
$$f(v) \equiv \max_{\{ \Omega\mid  {\rm U}(\Omega) = v\}} {\rm U}_{\varrho}(\Omega).$$
 Second, we consider Bell's operator employing general  measurements $\Omega'$ and decompose a general measurement operator   as $\hat{A}_{s_m} = r_{0|s_{m}} \vec{a}_{s_{m}} \cdot \vec{\sigma} + r^{*}_{s_{m}} \openone = \sum_{x_{m}=0, \pm} r_{x_{m}|s_{m}} \hat{A}_{x_{m}|s_{m}}$, with  $\vec{a}_{s_{m}}$ is unit vector and $\vec{\sigma}$ the Pauli vector, with  $\hat{A}_{0|s_{m}}=\vec{a}_{s_{m}} \cdot \vec{\sigma} $ and   $\hat{A}_{\pm|s_{m}}=\pm \openone$, and  coefficients $r_{\pm|s_{m}} \equiv \left(1 - r_{0|s_{m}} \pm r^{*}_{s_{m}}\right)/2$  and  $\sum_{x_{m}=0, \pm} r_{x_{m}|s_{m}}=1$  with  $0 \leq r_{x_{m}|s_{m}}$. Thus,  the Bell operator   can be decomposed into those involving projective measurements and identity operators  as
$$\hat{\beta} (\Omega')=t_{0}\hat{\beta}(\Omega_{0})+\sum_{i\neq 0}t_{i}\hat{\beta}(\Omega^{*}_{i})$$
where  $\Omega_{0}$  and $\Omega^{*}_{i}$ denote  the set of measurements  $ \{(\hat{A}_{x_{1}|0}, \hat{A}_{x'_{1}|1}); \cdots;   (\hat{A}_{x_{n}|0}, \hat{A}_{x'_{n}|1})\}$,  with  $(\hat{A}_{x_{m}|0}, \hat{A}_{x'_{m}|1} )$ are  the decomposition elements  of measurements  on $m$-th side.  The  subscript   $i=\{(x_{1}, x'_{1}); \cdots; (x_{n}, x'_{n})\}$  specifies  the involved  elements in $\Omega_{0}$ and $\Omega^{*}_{i}$, and  the relevant  Bell's operator has a weight as   $t_{i}=\Pi_{m} (r_{x_{m}|0} \cdot r_{x'_{m}|1})$.  Note that    $\hat{A}^{(n)}_{0|\nu}= \vec{a}^{(n)}_{\nu}\cdot \vec{\sigma}$  can be seen as a projection measurement,    by  $\Omega_{0}$ we specify measurement set where  all the measurements  involved are  rank one and projective, and  $\Omega^*_{i}$  specifies the measurement set contains  identity operators. We have 
 \begin{eqnarray}
\beta_{\varrho}(\Omega')&=&t_{0}{\beta}_{\varrho}(\Omega_{0})+\sum_{i\neq 0}t_{i}{\beta}_{\varrho}(\Omega^{*}_{i})\nonumber \\
&\le & t_{0}{\rm U}_{\varrho}(\Omega_{0})+\sum_{i\neq 0}t_{i}{\rm U}_{\varrho}(\Omega^{*}_{i}). \nonumber\\
&\leq  &  t_{0}{\rm U}_{\varrho}(\Omega_{0})+\sum_{i\neq 0}t_{i}\max_{\pm}{\rm U}_{\varrho}(\Omega_{i, \{\pm\}}) \nonumber  \\
&\leq  &  t_{0}f({\rm U}(\Omega_{0}))+\sum_{i\neq 0}t_{i}f(\max_{\pm}{\rm U}(\Omega_{i, \{\pm\}}))\nonumber \\
 &\leq & f[ t_{0}{\rm U}(\Omega_{0})+\sum_{i\neq 0}t_{i}\max_{\pm}{\rm U}(\Omega_{i, \{\pm\}})] \nonumber \\
& \leq & f(\beta_{\rho}). \label{basic}
 \end{eqnarray}  
 In the second inequality, we note that $\Omega^*_{i}$  contains, for instance,  $\hat{A}_{s_{m}} = \openone$ and rank one projection measurement $\hat{A}_{1-{s}_{m}}$. These two measurements commute with each other. To compute bounds  ${\rm U}(\Omega^{*}_{i})$ and ${\rm U}_{\varrho}(\Omega^{*}_{i}) $, we can replace $\hat{A}_{s_{m}}$ with  $\hat{A}_{s_{m}} = \pm \hat{A}_{1-{s}_m}$. This leads to  sets of measurement settings denoted $\Omega_{i, \{\pm\}}$.  We observe that (see A.1 in supplemental material (SM) in  \cite{supplemental})   $${\rm U}_{\varrho}(\Omega^{*}_{i}) \leq \max_{\pm} {\rm U}_{\varrho}(\Omega_{i, \{\pm\}})$$ In the third inequality of Eq.\ref{basic},  we have utilized the trade-off  between ${\rm U}_{\varrho}(\Omega)$ and ${\rm U}(\Omega)$, namely ${\rm U}_{\varrho}(\Omega) \leq f({\rm U}(\Omega))$, where the structure function $f$ is assumed to be monotonically decreasing and concave. It is also assumed to  hold 
 $$\max_{\pm} {\rm U}_{\varrho}(\Omega_{i, \{\pm\}}) \leq f\left( \max_{\pm} {\rm U}(\Omega_{i, \{\pm\}}) \right).$$  In the following examples,  this inequality is ensured by the fact that  these upper bounds are independent on the sign $\{\pm\}$.  The  last inequality is due to that      $\beta_{\rho}(\Omega')\leq  t_{0}{\rm U}(\Omega_{0})+\sum_{i\neq 0}t_{i}\max_{\pm}{\rm U}(\Omega_{i, \{\pm\}})$ and  $f$ is  a decreasing function. }

\subsection{ A. Strengthening Bell's test in bipartite systems }

 We begin with the minimal Bell scenario, where there are two parties with each party performing one of two binary measurements, \( \hat{A}_{\nu} \) and \( \hat{B}_{\mu} \), with \(\nu, \mu \in \{0, 1\}\). The relevant Bell inequality, namely, CHSH inequality that reads 
$$ {A_{0}B_{1}} + {A_{1}B_{0}} + {A_{0}B_{0}} - {A_{1}B_{1}}\leq 2 ,
$$
with \( {A_{\nu} B_{\mu}} = \mathrm{Tr}(\rho \hat{A}_{\nu} \hat{B}_{\mu}) \) denoting the expectation value.

 \begin{Theorem}
In the minimal Bell scenario, if the NLCG  measurements  that can generate  correlations  with a Bell value \( \beta_{\rho} > 2 \) on some state $\rho$, generate a correlation with the Bell value 
\begin{equation}
\beta_{\rho'} \geq \textstyle \sqrt{2 + 2\sqrt{1 - \left(\frac{\beta^2_{\rho}}{4} - 1\right)^2}} 
\end{equation} 
then the underlying two-qubit state \( \rho' \) can be certified to be entangled. 
 \end{Theorem}

First,  we  consider projective measurements \( \hat{A}_{\nu} = \vec{a}_{\nu} \cdot \vec{\sigma} \) and \(\hat{B}_{\mu} = \vec{b}_{\mu} \cdot \vec{\sigma} \) with \( |\vec{a}_{\nu}| = |\vec{b}_{\mu}| = 1 \), where \( \vec{\sigma} \) is the vector of Pauli matrices. In this case, \( a \equiv \vec{a}_{0} \cdot \vec{a}_{1} \) and \( b \equiv \vec{b}_{0} \cdot \vec{b}_{1} \) uniquely characterize the measurements up to a unitary operation. We shall specify these projective measurements by  \( \Omega \equiv \{a, b\} \). For general state  $\rho$ and a separable state $\varrho_{(11)}=\sum_{r}p_{r}\rho_{r}^{(A)}\otimes \rho^{(B)}_{r}$, we have the upper bounds of the CHSH inequality as ( {see  Supplemental Material \cite{supplemental} A.2})
  \begin{eqnarray}
\beta_{\rho}(\Omega)& \leq& 2\sqrt{1+\bar{a}\bar{b}}\equiv {\rm U}(\Omega), \label {twoen}\\
 \beta_{\varrho_{(11)}}(\Omega)&\le& \sum _{\pm}\textstyle \sqrt{1\pm \bar {a}\bar{b}}\equiv {\rm U}_{(11)}(\Omega) \label{twose}
\end{eqnarray}
 where $\bar c=\sqrt{1-c^{2}}$ . It is easy to see that these two bounds exhibit a trade-off   
\begin{eqnarray}
{\rm  U}_{(11)}{=}  \sqrt{2+2\sqrt{1-\left(\textstyle\frac{{\rm U}^{2}}4-1 \right)^{2}}}\equiv f({\rm U}), \label{trade}
\end{eqnarray}
 {where $f(\cdot)$ is {concave} and  decreasing function.  Here, $\max_{\pm} {\rm U}_{{(11)}}(\Omega_{i, \{\pm\}}) \leq f\left( \max_{\pm} {\rm U}(\Omega_{i, \{\pm\}}) \right)$  stands because  $\max_{\pm}{\rm U}_{(11)}(\Omega_{i, \{\pm\}})=2$   and    ${\rm U}(\Omega_{i, \{\pm\}})=2$ where  $\Omega_{i, \{\pm \}}$ could be  $\{a, \pm 1 \}$, $\{\pm1, b\}$, or $\{\pm 1, \pm 1\}$.   Then, for separable state and with the argument  Eq.(\ref{basic})  we have   
\begin{eqnarray}
\beta_{\varrho_{(11)}}(\Omega')\leq f(\beta_{\rho}). 
\end{eqnarray} Thus, if a Bell  test uses  NLCG  measurements $\Omega'$  that  can generate correlation with   Bell's value  \( \beta_{\rho} > 2 \), the bound for separable states from this Bell's test can be strengthened into \( f(\beta_{\rho})<2 \).  Any correlation exceeding this bound can certify entanglement.}
  {We emphasize that while NLCG measurements can detect entanglement through local correlations, they must themselves be not jointly measurable and able to  generate the nonlocal correlations as required by the definition.  }  

 {We note that the trade-off $f$ between the  ${\rm U}_{(11)}(\Omega)$ and ${\rm U}(\Omega)$ is due to the different behaviors of  $\rm U$ and ${\rm U}_{(11)}$ depending on the angles $|a|$ and $|b|$, namely, they    are respective  positively  or negatively  dependent  on that. Such a difference  may not be specific to a particular Bell inequality. In the SM \cite{supplemental}, we provide a numerical approach enabling one  to explore this point.}




\subsection{B. Tripartite Bell's test}
We now move to  multi-qubit system and begin with a tripartite system, for which a fully separable state can be written as  $\varrho_{(111)}=\sum_{r}p_{r}\rho_{r}^{(A)}\otimes \rho_{r}^{(B)}\otimes \rho^{(C)}_{r}$,  and  partially entangled states allow the decomposition     $\varrho_{(21)}=\sum_{r}p_{r}\rho^{(AB)}_{r}\otimes \rho^{(C)}_{r}+\sum_{r}p'_{r}\rho^{(BC)}_{r}\otimes \rho_{r}^{(A)}+\sum_{r}p''_{r}\rho_{r}^{(AB)}\otimes \rho_{r}^{(B)}\neq \rho_{(111)}$.  
States that cannot be written as \( \varrho_{(21)} \) and \( \varrho_{(111)} \) are genuinely tripartite entangled.  As the first example, we consider the tri-partite Mermin inequality {using NLCG measurements and tightening the bounds for separable states and partially entangled  states}.

{\bf Lemma} {\it 
In a $(3,2,2)$ scenario, adhere to the notions $\hat{A}_{\nu}=\vec{a}_{\nu}\cdot\vec{\sigma} $,   $\hat{B}_{\mu}=\vec{a}_{\mu}\cdot \vec{\sigma} $,  and  $\hat{C}_{\tau}=\vec{a}_{\tau}\cdot \vec{\sigma} $ with  $\nu, \mu, \tau\in \{0, 1\}$ we have  tri-partite Mermin's inequality~\cite{PhysRevLett.65.1838}  as
  \begin{eqnarray}
F_{3}\equiv ({A}_{0}B_{1}+A_{1}B_{0})C_{0}+(A_{0}B_{0}-A_{1}B_{1})C_{1}\leq 4. 
\end{eqnarray}
Given fixed projection  measurement settings specified by $\Omega=\{a,b,c\}$ with  $a=\vec{a}_{0}\cdot  \vec{a}_{1}$, $b=\vec{b}_{0}\cdot \vec{b}_{1}$, and  $c=\vec{c}_{0}\cdot \vec{c}_{1}$,  we have  the maximum Bell value ( {see SM \cite{supplemental} A.3} )  i)
  \begin{eqnarray}
 \beta_{\rho} \leq 2\sqrt{1+\bar a\bar b+\bar a\bar c+\bar b\bar c}\equiv  {\rm U}(\Omega).   \label {three}
\end{eqnarray}
 for a general state and ii)  
\begin{eqnarray}
 \beta_{\varrho_{(2 1)}} &\le &\max_{\Pi (a,b,c)}\sum_{\pm}\sqrt{1+\bar a\bar b\pm(\bar a\bar c+\bar b\bar c)} \nonumber \\
&\equiv &{\rm U}_{(2 1)}(\Omega) \label{twoone}
\end{eqnarray}
for a two-separable state it holds, with the  maximum taken over all permutations of  $(A,B,C)$ specified by  $\Pi (a, b, c)$, while iii) 
\begin{eqnarray}
\beta_{\varrho_{(111)}} & \le& \max_{u,\Pi (a, b,c) }\sum_\pm\sqrt{\frac{(1\pm \bar a \bar b )(1+c\cos u)+ab\bar c\sin u}2}\nonumber \\
&\equiv &{\rm U}_{( 111)}(\Omega ).\label{ones}\end{eqnarray}
for fully separable states.}

Relevant to the same $\Omega$, we have the trade-off between these upper bounds
 \begin{eqnarray}
{\rm U}_{(21)} \leq f_{(21)}({\rm U}), 
\end{eqnarray}
where 
$$f_{(21)}(v) \equiv \max_{\{\Omega\mid  {\rm U}(\Omega) = v\}} {\rm U}_{(21)}(\Omega)=\frac {v+\sqrt{16- v^{2}}}2$$  with   $2\sqrt{2} \leq v \leq 4$ (see A.3 in  {\cite{supplemental}}), and $2\sqrt{2}$ and $4$ are the maximum values that can be achieved using $\varrho_{(21)}$ and general state $\rho$ under  the possible projective measurements.


To generalize the above trade-off to the  general measurements case, we note that $\max_{\pm}{\rm U}_{{(21)}}(\Omega_{i, \{\pm\}})\leq f(\max_{\pm}{\rm U}(\Omega_{i, \{\pm\}}))$,  as these upper bounds  are independent on the signs of $a, b$, and  $c$.   As $f_{(21)}$ is monotonically decreasing  and concave,   with the argument Eq.(\ref{basic})  we have   
 \begin{eqnarray}
 \beta  _{\varrho_{(21)}} (\Omega')  \leq  f_{(21)}({\beta}_{\rho}). \label {sign}
\end{eqnarray} 
\begin{Theorem}
For tri-partite Bell's test featuring qubit systems, if the devices  can generate nonlocal correlations with a Bell value  \( \beta_{\rho} > 2 \sqrt{2}\) , then any correlation with the Bell value \( \beta_{\rho'} >  \beta_{\rho}/2+\sqrt{4-\beta^{2}_{\rho}/4}\) can certify that the relevant state \( \rho' \) is   genuine tripartite entangled. 
 \end{Theorem}

\begin{figure}
\begin{center}
\includegraphics[width=0.435\textwidth]{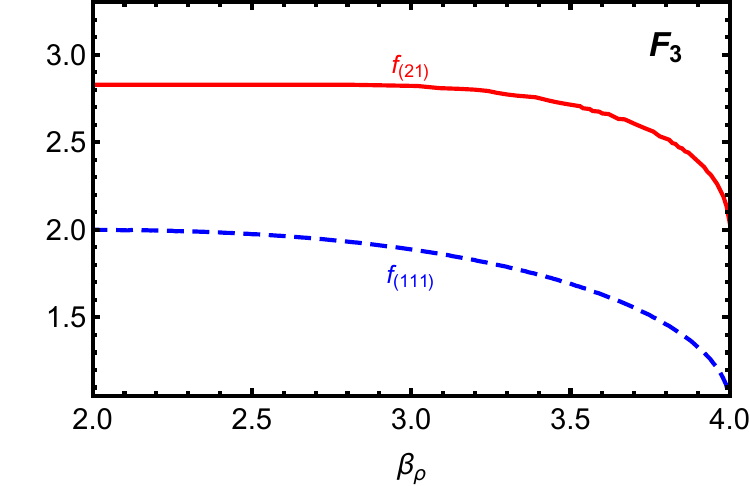}
\end{center}
\caption{For  Bell's inequality $F_3$, the tightened bounds   for bipartition separable state and fully separable state using Bell's values  $\beta_{\rho}$ are given. With devices that can generate a Bell value $\beta_{\rho}$, the strengthened  bound for bipartite separable  states is given as $f_{(21)}$ (red solid curve) and the bound for fully separable states is given  as $f_{(111)}$ (blue dashed curve), which coincide with the structure functions dealing with projective measurements. }\label{fig1}
\end{figure}
 
 {Similarly, in the case of fully separable state we have trade-off $\beta_{\varrho_{(111)}} \leq f_{(111)}(\beta_{\rho})$, in which although we do not have a closed form of the structure function $f_{(111)}$, its property of concavity and monotonically decreasing can be shown numerically by  Fig.1.     The inequality $\max_{\pm}{\rm U}_{{(111)}}(\Omega_{i, \{\pm\}})\leq f(\max_{\pm}{\rm U}(\Omega_{i, \{\pm\}}))$ follows from the fact that these upper bounds are  independent on the sign  (signs of $c$ and $ab$   in the expression can be absorb into the trigonometric functions.) }  For an instance, when NLCG measurements have ever generated some correlation with a Bell value $\beta_\rho=4$, any further correlations generated by the same set of NLCG {measurements}  with a Bell value $\beta_{\rho'}> 1$ can verify the entanglement of the underlying state $\rho'$ instead of  $\beta_{\rho'}>2$ as required by   standard Bell's tests.

\subsection{C. $n$-partite Bell's scenario}
Let us move to a general $(n, 2, 2)$ Bell scenario and consider Bell's inequalities defined using the MK polynomials, and then provide a systematic study for the strengthened bounds. MK polynomials are defined recursively as:
  \begin{eqnarray}
\hat{F}_{n}&=&\frac{\hat{A}_{n}+\hat{A}'_{n}}2\hat{F}_{n-1}+\frac{\hat{A}_{n}-\hat{A}'_{n}}2\hat{F}'_{n-1}, \nonumber \\
\hat{F}'_{n}&=&\frac{\hat{A}_{n}+\hat{A}'_{n}}2\hat{F}'_{n-1}- \frac{\hat{A}_{n}-\hat{A}'_{n}}2\hat{F}_{n-1},
\end{eqnarray}
where  $\hat{A}_{m}, \hat{A}'_{m}$ are the two local measurements on $m$-th side, with $\hat{F}_{2}=\hat{A}_{1}\hat{A}'_{2}+\hat{A}'_{1}\hat{A}_{2} + \hat{A}_{1}\hat{A}_{2}-\hat{A}'_{1}\hat{A}'_{2}$  and   {$\hat{F}'_{2}$}$=\hat{A}_{1}\hat{A}'_{2}+\hat{A}'_{1}\hat{A}_{2}  -\hat{A}_{1}\hat{A}_{2}+\hat{A}'_{1}\hat{A}'_{2}$. Specifically, when $n=2$, we have the CHSH inequality as  $\hat{F}_{2}\leq 2$. When $n=3$, we have the operator for  {Mermin's} inequality  as  $\hat{F}_{3}$.  
 We give analytic bounds for general state and bipartition separable state as  { (see  A.4 in \cite{supplemental}) . Numerical approach for general Bell's inequality is also discussed in  A.5. in \cite{supplemental}} 
 \begin{Theorem}
 For a general n-qubit state on which projective measurements $\hat{A}_{m} = \vec{a}_{0|m} \cdot \vec{\sigma}$  and   $\hat{A}'_{m} = \vec{a}_{1|m} \cdot \vec{\sigma}$ are performed with $\Omega=\{a_{m}=\vec{a}_{0|m}\cdot \vec{a}_{1|m}\}$ denoting measurement settings, it holds the following upper bound
 \begin{eqnarray}
F_n \cos t+F_n'\sin t  \le 2\sqrt{\delta_n+\epsilon_n\sin 2t}\equiv {\rm U}^{(t)}\label {nt}
\end{eqnarray}
where $2\delta_{l}=\Pi^{l}_{m=1}(1+\bar{a}_{m})+\Pi^{l}_{m=1}(1-\bar{a}_{m})$ and $\epsilon_l=\Pi^{l}_{m=1}a_{m}$ for arbitrary value of $t$.  
For bi-partition  separable state $\varrho_{(kk')}=\sum_{r}p_{r}\rho^{(r)}_{k}\otimes \rho^{(r)}_{k'}$ with partition $k+ k'=n$, we have the upper bound 
\begin{eqnarray*}
 \beta_{(k k')}=F_n\cos t+F_n'\sin t \leq {\rm U}^{(t)}_{(kk')},\label{kkk} 
\end{eqnarray*}
where  ${\rm U}^{(t)}_{(kk')}\equiv \sqrt{2\alpha_{t}+2\sqrt{\alpha^2_{t}-(\delta_k^2-\epsilon_k^2)(\delta_{k'}^2-\epsilon_{k'}^2)}}$ and $\alpha_{t}\equiv (\delta_{k'} \delta_k,+\epsilon_{k'}\epsilon_k\sin 2t)$.

 \end{Theorem}
The technique  used here allows one to tighten the  bounds for  the state allowing   a more complex multiplicative decomposition, but the complexity goes beyond the scope of this paper and is not provided here.  As the first example, we consider a system of two qubits and the {tilted} CHSH inequality  $(\textstyle {\sin t+\cos t }) (A_{0}A'_{1}+A'_{0}A_{1})/2 +(\textstyle {\sin t-\cos t })(A_{0}A_{1}-A'_{0}A'_{1})/2=\sin t F_{2}+\cos t F'_{2}$, for which, we have the upper-bound for a  general state $2\sqrt{1+\bar{a}_{1}\bar a_{2}+a_{1}a_{2}\sin 2t}$ and for a general separable state $\sum_{\rm \pm} \sqrt{1+a_{1}a_{2}\sin 2t\pm \bar {a}_{1} \bar {a}_2} $. When $t=\frac{\pi}{2}$, we recover Eq.(\ref{twoone}) and Eq.(\ref {twose}). 
As another example, we consider tripartite Bell's quantity  ${F}_{3}\sin t  +{F}'_{3}\cos t  $, the  bounds for general state and partially entangled state follow directly  from Eq.(\ref{nt})  and Eq.(\ref{kkk}), while the bounds for fully separable states can be derived through an analogous approach ( {see A.3 in \cite{supplemental}} ).    {Specifically, we provide the tightened bounds when  $\sin t=\cos t$, namely,  $F_3+F_3'\leq 4$   in  A.5. in \cite{supplemental} }

 {These tightened  bounds are practically  useful. Recall that entanglement is typically detected using trusted devices. To verify  whether  these devices function as intended, they are often verified  with  extreme correlations that are specific to them. Our approach offers an advantage: it does not require the generation of such extreme correlations, which can be experimentally demanding. The presence of nonlocal correlations is sufficient to tighten the bounds on Bell violations, with which, one can detect entanglement with correlations that are not necessarily nonlocal.   }

\section{III. NPA approach to strengthen Bell's inequalities}

Let us consider the Bell's scenario  where some parties have certain information about the measurements. We  now adapt  the NPA hierarchy to incorporate the information of measurements. The NPA hierarchy \cite{PhysRevLett.98.010401} is a sequence of tests designed to characterize correlations that can be realized with a quantum system ( {see Sec. B in \cite{supplemental}}).  Here, we use the NPA hierarchy to help the detection of entanglement while not involving Bell's inequality.

Here, we consider a scenario in which Alice uses some known measurement settings on her qubit and can incorporate them  into the NPA hierarchy. Define $\mathcal{O}$ as a set of $\{O^{(A)}_e \otimes O^{(B)}_f\}$, with $O^{(A)}$ specifying the measurement operators  relevant to Alice's system, their products, or linear combinations, and $O^{(B)}$ for Bob's side. Then, one can define a  matrix $\Gamma$ with entries  
$\Gamma_{ee', ff'} = {\rm Tr}[\rho (O^{(A)\dagger}_e O^{(A)}_{e'} + O^{(A)\dagger}_{e'} O^{(A)}_e)\otimes(O^{(B)\dagger}_f O^{(B)}_{f'} + O^{(B)\dagger}_{f'} O^{(B)}_f)],$
where $ee'$ and $ff'$ specify the index of lines and columns of this entry in the matrix. When $O^{(A)}_e, O^{(A)}_{e'}$ and $O^{(B)}_f, O^{(B)}_{f'}$ involve commuting operators, the corresponding matrix entry $\Gamma_{ee', ff'}$ is relevant to the statistics from the experiment and  called physical terms and otherwise is referred to as nonphysical terms when  non-commuting operators are involved.

Consider a product state $\varrho=\rho^{(A)}\otimes \rho^{(B)}$, we have factorized matrix $\Gamma=\Gamma^{(A)}\otimes\Gamma^{(B)}$ where $\Gamma^{(A)}_{ee'}={\rm Tr}[\rho^{(A)} (O^{(A)\dagger}_{e} {O}^{(A)}_{e'}+O^{(A)\dagger}_{e'} {O}^{(A)}_{e} )]$ and  $\Gamma^{(B)}_{ff'}={\rm Tr}[\rho^{(B)}(O^{(B)\dagger}_{f} {O}^{(B)}_{f'}+O^{(B)\dagger}_{f'} {O}^{(B)}_{f})]$ ( {see section B. in \cite{supplemental}}),  and $\Gamma^{(A)}$ and $\Gamma^{(B)}$ are semidefinite. As a result,  
for a general separable state  $\varrho=\sum_{r}p_{r}\rho^{(A)}_{r}\otimes \rho^{(B)}_{r}$, the convex combination of the relevant   matrix  $\sum_{r}p_{r}\Gamma_{r}$ is also positive. This  may not be the case for  entangled states.  For separable  states, one can assign some real numbers  to these nonphysical terms such that the matrix  is semidefinite. Otherwise, the state is entangled.    Focusing on qubit systems, for instance, measurements on Alice's side are known as  \( O^{(A)}_{e}  = r_{0} \vec{a}_{0} \cdot \vec{\sigma}+r^{*}_{0}\openone \) and \( O^{(A)}_{e'} =  r_{1} \vec{a}_{1} \cdot \vec{\sigma}+r^{*}_{1}\openone \),   we have 
$O^{(A)\dagger}_e O^{(A)}_{e'} + O^{(A)\dagger}_{e'} O^{(A)}_e= 2 \left( r^{*}_{0} O^{(A)}_{e'} + r^{*}_{1}O^{(A)}_{e} + (r_{0} r_{1} a + r^{*}_{0} r^{*}_{1}) \openone \right),$
which entirely consists  of operational terms and constants.  {Then, one can express these non-physical terms with  this  expectation values from experiment and  with known parameters relevant to measurements. If  potential state is separable, the matrix must be positive.    Otherwise the matrix is negative, the potential state must be entangled. In  section B in \cite{supplemental}, we provide an example showing that local correlation can verify entangled state. }

\section{Discussion and Conclusion}

Bell tests are widely employed for entanglement detection. When the devices in a Bell test are only partially characterized, the bounds for separable states may deviate from standard local bounds. To improve the efficiency of these methods, this paper explores how to strengthen Bell inequalities for separable states in multi-qubit systems. For bipartite and tripartite systems, we derive explicit bounds for both separable and general quantum states, revealing inherent trade-offs between them. These results enable entanglement detection beyond local bounds using NLCG measurements, i.e., measurements  capable  of generating nonlocal correlation. Furthermore, we introduce a systematic approach to strengthen Bell inequalities defined in terms of  Mermin-Klyshko (MK) polynomial for \textit{n}-qubit systems, obtaining analytical bounds for various types of separable states. Additionally, in scenarios where some measurements are well-characterized, we demonstrate how the NPA hierarchy can be applied to strengthen  bounds for separable states.

\subsection*{Acknowledgment}
We appreciate the discussion with Armin Tavakoli.  L.S.\ is supported by Key-Area Research and Development Program of Guangdong Province Grant No. 2020B0303010001.
S.Y.\ is supported by Innovation Program for Quantum Science and Technology (2021ZD0300804) and Key-Area Research and Development Program of Guangdong Province Grant No. 2020B0303010001.  

\bibliographystyle{plain}
\bibliography{nonlocality}

\appendix
\onecolumngrid

\section{A. Tighten the bound of Bell's inequality for separable state}

\subsection{A.1.  The decomposition  of  Bell's  value  for general  measurements}

 {In the Eq.(1) main-text, the decomposition of Bell's test often involves  the measurement setting $\Omega^*_{i}$, which  contains  some identities as measurements.  For the purpose of computing the upper bound of Bell's operator  employing  $\Omega^{*}_{i}$,  one can replace the identity in the Bell test with projection measurements compatible with another measurement on the same side, which leads to   measurement settings $\Omega_{i, \{\pm\}}$. We now show   upper bounds (for general state or separable state)  
$${\rm U}(\Omega^*_{i}) \leq \max_{\pm} {\rm U}(\Omega_{i, \{\pm\}}) .$$
For simplicity, we first show that in  the minimal Bell's scenario. }

\emph{Case 1: }
When a measurement setting  $\Omega^*_{i}=\{\{\hat{A}_{0}, \hat{A}_{1}\}, \{\hat{B}_{0}, \hat{B}_{1}\}\}$ contains  one  identity on Alice's side, for instance,  $\hat{A}_{0}=\openone$ and   $A_{1}=\{|0 \rangle \langle 0|_{1}, |1 \rangle \langle 1|_{1}\}$. We have
\begin{eqnarray}
A_{1}{B}_{\mu}&=&{\rm Tr}(\rho \hat{A}_{1}\hat{B}_{\mu})={\rm Tr}(\rho |0\rangle \langle 0|_{1} \hat{B}_{\mu})- {\rm Tr}(\rho |1\rangle \langle 1|_{1} \hat{B}_{\mu})\nonumber\\
&=&p(0|A_{1}){\rm Tr}( \rho^{(B)}_{0|A_{1}} \hat{B}_{\mu})-p(1|A_{1}){\rm Tr}(\rho^{(B)}_{1|A_{1}} \hat{B}_{\mu})\nonumber\\
&=&\sum_{a'=0, 1}p(a'|A_{1}){\rm Tr}(|a'\rangle \langle a'|_{1}\otimes  \rho^{(B)}_{a'|A_{1}} \hat{A}_{1}\hat{B}_{\mu})\nonumber\\
&=&\sum_{a'=0, 1}p(a'|A_{1}){\rm Tr}(\rho_{1_{a'}} \hat{A}_{1}\hat{B}_{\mu}), 
\end{eqnarray}
where   $p(a'|A_{1})\equiv {\rm Tr}(\rho |a'\rangle \langle a'|_{1}\otimes \openone_{B} )$  and $a'=0,1$  specifies outcomes,  and $\openone_{B}$ is   identity acting on Bob's system, and   $\rho^{(B)}_{a'|A_{1}}={\rm Tr}_{A}(\rho| a'\rangle \langle a'|_{1})$ with the partial trace taken on Alice's side, and  $\rho_{1_{a'}}= |a'\rangle \langle a'|_{1}\otimes  \rho^{(B)}_{a'|A_{1}}$,
and \begin{eqnarray}
{{A_{0}B_{\mu}}}={\rm Tr}(\rho \hat{A}_{0}\hat{B}_{\mu}) =p(0|A_{1}){\rm Tr}(\rho_{1_{0}} \hat{A}_{1}\hat{B}_{\mu})+p(1|A_{1}){\rm Tr}(\rho_{1_{1}} (-\hat{A}_{1})\hat{B}_{\mu}).   \label{openonea}
\end{eqnarray}
Then, we have  an equivalent decomposition of this Bell's value   $\beta(\Omega^{*}_{i})$ as 
\begin{eqnarray}
\beta(\Omega^{*}_{i})=p(0|A_{1})\beta_{\rho_{1_{0}}}(\Omega_{i, +})+  p(1|A_{1})\beta_{\rho_{1_{1}}}(\Omega_{i, -})\leq \max_{\pm } {\rm U}(\Omega_{i,\{\pm \}}).  
\end{eqnarray}
where $\Omega_{i, \pm}=\{\pm, \{\hat{B}_{0}, \hat{B}_{1}\}\}$ means measurement settings, where  Alice uses compatible and rank one projection measurements. 

\emph{Case. 2: }  When there are two identities in $\Omega^*_{i}=\{\{\hat{A}_{0}, \hat{A}_{1}\},   \{\hat{B}_{0}, \hat{B}_{1}\}\}$, for instance $\hat{A}_{0}=\hat{A}_{1}=\openone$ on Alice's side,  such  a correlation can be equivalently generated using $\hat{A}_{0}=\hat{A}_{1}=\{|0\rangle \langle 0|, |1\rangle \langle 1|\}$  and state $\rho^{*}=|0\rangle \langle 0 |\otimes {\rm Tr}_{A}(\rho)$ as
\begin{eqnarray}
A_{\nu}B_{\mu}={\rm Tr }(\rho \hat{A}_{\nu} \hat{B}_{\mu} )={\rm Tr}(\rho^{*} \hat{A}_{\nu} \hat{B}_{\mu} ), 
\end{eqnarray}
 Finally, we have in this setting,  $$\beta_{\rho}(\Omega^{*}_{i})=\beta_{\rho^{*}}(\Omega_{i, +})\leq \max_{\pm}{\rm U}(\Omega_{i, \pm}),$$
for any state $\rho$. 

This analysis can be performed independently on each side when  measurements $(\pm )\openone $ are involved as well as being  extended to the general $(n, 2, 2)$ scenario via a one-by-one decomposition of measurements.   We  then have   
$$\beta_{\rho}(\Omega^*_{i})\leq {\rm U}(\Omega^*_{i}) \leq  \max_{\pm} {\rm U}(\Omega_{i, \{\pm\}}) .$$

\subsection{A.2. Bound of Bell's inequalities using projection measurements}

\subsubsection{Two-partite Bell's Inequality}
In this section, we give bounds for  states presenting various structures of entanglement.  Our basic tool is an uncertainty relation for two protective measurements  $\hat{A}_{\nu}=\vec a_\nu\cdot\vec \sigma $ with unit Bloch vector $\vec a_\nu$ with $\nu=1,2$. Clearly,  two operators $\hat{A}_{\pm}=\hat{A}_{0}\pm\hat{A}_{1}$ are anti-commuting and as 
$  
 \hat{A}^{2}_{\pm}=2(1\pm a)$ with $ a=\vec a_0\cdot\vec a_1$
 we have
\begin{eqnarray}
\frac{\langle \hat A_{0} + \hat A_{1}\rangle^{2}}{2(1+a)}+\frac{\langle \hat A_{0} - \hat A_{1} \rangle^{2}}{2(1-a)}\leq 1. \label{UR}
\end{eqnarray}
Thus there exist $|r|\leq 1$, $0\leq u_{1}< 2\pi$ such that 
\begin{eqnarray}
 \langle \hat A_{+}\rangle &\equiv& r\sqrt{2(1+a)}\cos v_{1}, \label{singl1} \\
 \langle \hat A_{-}\rangle&\equiv& r\sqrt{2(1-a)}\sin v_{1}. \label{singl2}
\end{eqnarray}

We first consider the  Bell scenarios $(2,2,2)$ and $(3,2,2)$ as examples, and observers are labeled with $A, B, C$. Each observer performs 
 two projection measurements and, without loss of generality, we assume the first measurement $\hat A_0, \hat B_0, \hat C_0, $ to be $\sigma_z$ while the second measurement to be $\hat{A}_{1}=a\sigma _{z}+\bar{a}\sigma_{x}$, $\hat{B}_{1}=b\sigma_{z}+\bar{b}\sigma_{x}$, and $\hat{C}_{1}=c\sigma_{z}+\bar{c}\sigma_{x}$ with $\bar e=\sqrt{1-e^2}$. 
\begin{Lemma} 
Given a fixed set of projective local measurements $\Omega=\{ a, b\}$, for CHSH inequality
 $$\beta_{\rm CHSH}\equiv A_{0}B_{1} +A_{1}B_{0} +A_{0}B_{0} -A_{1}B_{1}, $$
the maximal Bell value is 
  \begin{eqnarray}
\beta_{\rm CHSH}\leq 2\sqrt{1+\bar{a}\bar{b}}, \label {twoen2}
\end{eqnarray} 
for any quantum state while 
\begin{eqnarray}
\beta^{(sep)}_{\rm CHSH}\le \sqrt{1+\bar a\bar b}+\sqrt{1-\bar a\bar b}. \label{twose}
\end{eqnarray}
for a general separable  state. 
 Roughly,  more incompatibility (a larger $\bar a\cdot\bar b$ ) implies a greater quantum bound while a smaller bound for separable states.
\end{Lemma}
\begin{proof}
     By denoting
\begin{eqnarray}
\hat{X}=\hat{A}_{0}\hat{B}_{1}+\hat{A}_{1}\hat{B}_{0}, \quad
\hat{Y}=\hat{A}_{0}\hat{B}_{0}-\hat{A}_{1}\hat{B}_{1}
\end{eqnarray}
and $m^{\pm}_{ab}\equiv \bar a\bar b\pm ab$, we see that $\{\hat{X}, \hat{Y}\}=0$ and 
  \begin{eqnarray*}
\hat{X}^2&=&2(1+ab)\openone+2\bar a\bar b\sigma^{(A)}_{y}\sigma^{(B)}_{y}\le 2(1+m_{ab}^ +)\openone,\\
\quad \hat{Y}^2&=&2(1-ab)\openone+2\bar a\bar b\sigma^{(A)}_{y}\sigma^{(B)}_{y}\le 2(1+m_{ab}^-)\openone.\label{sqrtwo}
\end{eqnarray*}
As a result, for a general two-qubit state $\rho$, it holds the following uncertainty relation
  \begin{eqnarray*}
\frac{X^2}{2(1+m^{+}_{ab})}+\frac{{Y}^{2}}{2(1+m^{-}_{ab})}
\le 1, \label{twoqubit}
\end{eqnarray*}
where and henceforth notion  $A=\langle \hat{A} \rangle$ will be often used   and it follows that $$ X=r_{2}\sqrt{2(1+m^{+}_{ab})}\cos u_{1},\
Y=r_{2}\sqrt{2(1+m^{-}_{ab})}\sin u_{1},$$
where $|r_{2}|\leq 1$, $0\leq v_{1}\leq 2\pi$. 
Then  CHSH inequality for a general quantum state  is bounded as 
  \begin{eqnarray}
\beta_{\rm CHSH}&=&X+Y \nonumber \\ &=&r_{2}\sqrt{2(1+m^{+}_{ab})}\cos u_{1}+
r_2\sqrt{2(1+m^{-}_{ab})}\sin u_{1} \nonumber \\
&\leq& r_{2}\sqrt{2[(1+m^{+}_{ab})+(1+m^{-}_{ab})][(\cos^{2}u_{1}+\sin^2u_{1})]}\nonumber \\
&\leq & 2\sqrt{1+\bar{a}\bar{b}},
\end{eqnarray}
where we have Schwarz inequality  in the first inequality. 
Specifically, when measurements are maximum incompatible, namely, $a=b=0$, we recover the maximum quantum violation $2\sqrt{2}$.

Now we consider  a  product state $\rho=\rho^{(A)}\otimes\rho^{(B)}$, it holds
   \begin{eqnarray*}
X&=&\langle \hat{A}_{0}\rangle\langle \hat{B}_{1}\rangle+\langle \hat{A}_{1}\rangle\langle \hat{B}_{0}\rangle=\frac{A_{+}B_{+} -A_{-}B_{-}}2, \nonumber \\
Y&=&\langle \hat{A}_{0}\rangle\langle \hat{B}_{0}\rangle-\langle \hat{A}_{1}\rangle\langle \hat{B}_{1}\rangle=\frac{ A_{+}B_{-}+A_{-}B_{+}}2, \nonumber
\end{eqnarray*}
where $\hat{A}_{\pm }\equiv \hat{A}_{0}\pm \hat{A}_{1}$ and $\hat{B}_{\pm }\equiv \hat{B}_{0}\pm \hat{B}_{1}$.
For a general angle $t$, we denote $B_t=\langle \hat B_+\cos t+\hat B_-\sin t\rangle$ and
$B_t'=\langle \hat B_+\sin t-\hat B_-\cos t\rangle$, we have 
\begin{eqnarray*}\label{f2}
&&X\cos t+Y\sin t=
\frac{{A}_+{B_t}+A_-B_t'}2\\
&=&\frac{{A}_+}{\sqrt{2(1+a)}}\cdot \sqrt{\frac{1+a}2}B_t+\frac{ {A}_-}{\sqrt{2(1-a)}}\cdot \sqrt{\frac{1-a}2}B_t'\\
&\le&\sqrt{\frac{1+a}2B_t^2+\frac{1-a}2B_t'^2}\\
&=&\sqrt{\frac{1+a\cos 2t}2{B}_+^2+\frac{1-a\cos 2t}2{B}_-^2+a{B}_+{B}_-\sin 2t},
\end{eqnarray*}
where  Schwarz inequality and Eq.(\ref{UR}) are used in the  inequality. 
From the uncertainty relation for a single qubit it follows that there exists $0\le r_2\le 1$ and angle $v_2$ such that $B_+=r_2\sqrt{2(1+b )}\cos v_2$ and $B_-=r_2\sqrt{2(1-b)}\sin v_2$. As a result
\begin{eqnarray*} \nonumber
&&(X\cos t+Y\sin t)^2\\ \nonumber
&=&r_2^2(1+ab\cos 2t+(a\cos 2t+b)\cos 2v_2 +a\bar b\sin 2v_2\sin 2t)\\ \nonumber
&\le&1+ab\cos 2t+\sqrt{(a\cos 2t+b)^2+a^2\bar b^2\sin^2 2t}\\ \nonumber
&=&{1+ab\cos 2t+\sqrt{(1+ab\cos 2t)^2-\bar a^2\bar b^2}} \\ 
&=&\Big(\textstyle\sum_\pm\sqrt{\frac{1+ab\cos 2t\pm\bar a\bar b}2}\Big)^2,
\end{eqnarray*}
where the inequality above is due to $r_2\leq 1$ and  the Schwartz inequality. As a result
\begin{eqnarray}
X\cos t+Y \sin t\le
\sum_\pm\sqrt{\frac{1+ab\cos 2t\pm\bar a\bar b}2}\label{22}
\end{eqnarray}

By taking $t=\frac\pi4$ we have proven Eq.(\ref{twose}) in the case of product state and thus for all separable states. Specifically,  when measurements are compatible $a=b=1$, we have the classical bound.  
\end{proof}
\subsection{A.3 Tri-partite Bell's inequality. }
\begin{Lemma}
In a $(3,2,2)$ scenario, given fixed measurement settings $a, b, c$,  Mermin's inequality reads 
$$\beta_{\rm M}\equiv (A_{0}B_{1}+A_{1}B_{0})C_{0}+(A_{0}B_{0}-A_{1}B_{1})C_{1}$$
has the maximum Bell value i)
  \begin{eqnarray}
\beta^{(3)}_{\rm M}\leq 2\sqrt{1+\bar a\bar b+\bar a\bar c+\bar b\bar c}. \label {three}
\end{eqnarray}
 for a general state and ii)
\begin{eqnarray}
\beta^{(21)}_{\rm M}\le \max_{\Pi (a,b,c)}\sum_{\pm}\sqrt{1+\bar a\bar b\pm(\bar a\bar c+\bar b\bar c)}\label{twoone}
\end{eqnarray}
for a two-separable state it holds, with maximum taken over all cyclic $(a,b,c)$ and $\Pi$ specifies the permutation, while iii) 
\begin{equation}
\beta^{(111)}_{\rm M}\le \max_{u, \Pi (a,b,c)}\sum_\pm\sqrt{\frac{(1\pm \bar a \bar b )(1+c\cos u)+ab\bar c\sin u}2}.\label{ones}
\end{equation}
for fully separable states.

\end{Lemma}

{\it Proof. } i)  Consider the following pair of Bell operators
\begin{eqnarray*}
\hat{F}_3&=&(\hat{A}_{0}\hat{B}_{1}+\hat{A}_{1}\hat{B}_{0})\hat{C}_{0}+(\hat{A}_{0}\hat{B}_{0}-\hat{A}_{1}\hat{B}_{1})\hat{C}_{1},\\
\hat{F}_3'&=&(\hat{A}_{0}\hat{B}_{1}+\hat{A}_{1}\hat{B}_{0})\hat{C}_{1}-(\hat{A}_{0}
\hat{B}_{0}-\hat{A}_{1}\hat{B}_{1})\hat{C}_{0}.
\end{eqnarray*}
we have $\hat{X}\hat{Y}=-\hat{Y}\hat{X}$ and thus
\begin{eqnarray*}
\hat{F}_3^2&=&(\hat{X}\hat{C}_{0}+\hat{Y}\hat{C}_{1})^2\\
&=&\hat{X}^2\hat{C}_{0}^2+\hat{Y}^{2}\hat{C}_{1}^{2}+\hat{X}\hat{Y}\hat{C}_{0}\hat{C}_{1}+\hat{Y}\hat{X}\hat{C}_{1}\hat{C}_{0}\\
&=& 4(\openone+\bar a\bar b\sigma^{(1)}_{z}\sigma^{(2)}_{z})+\hat{X}\hat{Y}[\hat{C}_{0}, \hat{C}_{1}] \nonumber \\
&=& 4(\openone +\bar a\bar b\sigma^{(1)}_{z}\sigma^{(2)}_{z})+([\hat{A}_{1},\hat{A}_{0}]+[\hat{B}_{1}, \hat{B}_{0}])[\hat{C}_{0} \hat{C}_{1}]\nonumber \\
&=& 4(\openone +\bar a\bar b\sigma^{(1)}_{z}\sigma^{(2)}_{z}+ \bar a\bar c\sigma^{(1)}_{z}\sigma^{(3)}_{z}+\bar b\bar c\sigma^{(2)}_{z}\sigma^{(3)}_{z}),
\end{eqnarray*}
 {where $X,Y$ deal with measurements on the first and second partite and $C$ deal with the third partite.} They are commuting.  It follows
\begin{eqnarray*}
\langle \hat{F}_3\rangle  \leq 2\sqrt{1+\bar a\bar b+\bar a\bar c+\bar b\bar c}. 
\end{eqnarray*}

ii) For a (21)-separable state, e.g., in a partition $AB|C$, namely,  $\rho=\rho^{(AB)}\otimes \rho^{(C)}$ we have
\begin{eqnarray*}
F_3&=&XC_{0}+YC_{1}\\
&\leq &\sqrt{2(1+c_{ab}^+) C^2_{0}+2(1+c_{ab}^-)C_{1}^2}\\
&=&\sqrt{2(1+\bar a\bar b) (C^2_{0}+C^{2}_{1})+2ab(C_{0}^2-C_{1}^2)}\\
&\le&\sqrt{2(1+\bar a\bar b)(1+c\cos 2v_3)+2ab\bar c\sin 2v_3}\\
&\le &\sqrt{2(1+\bar a\bar b)+2\sqrt{(c+\bar a\bar b c)^2+a^2b^2\bar c^2}}\\
&=&\sqrt{2(1+\bar a\bar b)+2\sqrt{(1+\bar a\bar b )^2c^{2}_{3}+a^2b^2\bar c^2}}\\
&=&\sum_\pm\sqrt{1+\bar a\bar b)\pm(\bar a+\bar b)\bar c}\\
&:=&\beta^{(21)}_{\rm M}(ab|c),
\end{eqnarray*}
where we have used  Eq.(\ref{sqrtwo}) and Schwarz inequality  in the first inequality, and single-qubit uncertainty relation in the second inequality and Schwarz inequality  in the third inequality.  As the upper bound $\beta_t(ab|c)$ is independent of the underlying state, it holds true also for a general $AB|C$ state.  

iii) For a full separable state $\rho=\rho_A\otimes\rho_B\otimes\rho_C$ we have
\begin{eqnarray*}
F_3&=&X C_{0}-Y C_{1}=R\cdot \left(X\cos\alpha+Y\sin\alpha\right)\le R \sum_\pm\sqrt{\frac{ 1\pm \bar a\bar{b}+ab\cos2\alpha}2},
\end{eqnarray*}
where
$R=\sqrt{ C_{0}^2+ C^2_{1}},\quad \cos\alpha=\frac{ C_0}{R},\quad \sin\alpha=\frac{ C_1}{R}$, and we have used 
   Eq.(\ref{22})   in the first inequality.  It follows from the constraint of uncertainty relation Eq(\ref{singl1}, \ref{singl2})  that 
\begin{eqnarray}
R^{2}&=&C^{2}_{0}+C^{2}_{1}=r^{2}_{3} ( 1+c\cos 2v_3),\\
R^2\cos 2\alpha &=&C_{0}^2- C^2_{1} 
= r^{2}_{3}\bar c \sin 2v_3. 
\end{eqnarray}
When the maximum value is obtained  $r_{3}=1$,   we finally arrive at Eq.(\ref{ones}).  

\subsubsection{Trade-off between  ${\rm U}$ and  ${\rm U}_{(21)}$}
First we consider a Bell's value that can verify genuine tripartite entanglement. We first  give the trade-off between  $\rm U$ and  ${\rm {U}}_{(21)}$.   By submitting    $\bar a\bar c+\bar b\bar c={{\rm U}^{2}}/{4}-\bar{a}\bar{b}-1$ with ${{{\rm U}^2}}/{4}\geq \bar{a}\bar{b}+1$  into Eq.(\ref{twoone}),   we have 
\begin{eqnarray}
f_{(21)}(v)&\equiv&\max_{{\rm U}(\Omega)=v} {\rm U}_{(21)}(\Omega)\nonumber \\
                         &= & \textstyle \max_{{\rm U}(\Omega)=v}\{\frac{v}{2}+\sqrt{2(1+\bar a\bar b)-\frac{v^{2}}{4}} \}\nonumber \\
 &= & \textstyle \max_{{\rm U}(\Omega)=v}
\{\frac{v}{2}+\sqrt{4-\frac{v^{2}}{4}} \}\nonumber \\
 &=&  \textstyle \frac{v}{2}+\sqrt{4-\frac{v^{2}}{4}},
\end{eqnarray}
where the maximum in the second equality  is taken when $1+\bar{a}\bar{b}=2$.   We do not obtain a closed form for $f_{(111)}. $ 
  {In the following we have shown these two structure functions}

\subsection{A.4   $n-$partite Bell's inequalities. }

\subsubsection{Preparation}

In general, the MK polynomial is defined recursively   with  $\hat{F}_n=\hat{F}_{n-1}\frac{\hat{A}_n+\hat{A'}_n}2+\hat{F}_{n-1}'\frac{\hat{A}_n-\hat{A'}_n}2,\quad \hat{F}_n'=\hat{F}_{n-1}'\frac{\hat{A}_n+\hat{A}'_n}2-\hat{F}_{n-1}\frac{\hat{A}_n-\hat{A}'_n}2$, where  $\hat{A}_{n}$ and $\hat{A}'_{n}$ are measurements on $n-$th side.  Specifically, when $n=2$, we have  $\hat{F}_2=\hat{X}+\hat{Y},\quad \hat{F}_2'=\hat{X}-\hat{Y}$, and $\hat{F}'_{2}$ is the  CHSH inequality.  In the following, we provide a universal method for computing the bounds for   $\hat{F}_{n}+\hat{F}'_{n}$,  $\hat{F}_{n}(\hat{F}'_{n})$  relevant to various kinds of entangled states. 

{\bf  Lemma }{The recursively  defined quantities 
$$\hat{F}_{n, \pm}=\frac{\hat{F}_n\pm \hat{F}'_n}{2}$$  Lemma. It holds 
\begin{eqnarray}
\hat{F}_n^2=\hat{F}_n'^2=4\hat \delta_n^+,\quad \hat{F}_n\hat{F}_n'=4(\epsilon_n+i\hat \delta_n^-),\quad \hat{F}_{n,\pm}^2=2(\hat\delta_n^+\pm \epsilon_n)
\end{eqnarray}
where, with $\hat{A}_n\hat{A}_n'=a_n+i\bar a_n\hat{Y}_n$, 
$$\hat \delta_n^\pm=\frac 12\prod_{j=1}^n(1+\bar a_j\hat{Y}_j)\pm\frac12\prod_{j=1}^n(1-\bar a_j\hat{Y}_j),\quad \epsilon_n=\prod_{j=1}^n a_j$$}

\begin{proof}
\begin{itemize}
\item Proof  of  $\hat{F}_n^2=\hat{F}_n'^2=4\hat \delta_n^+$

 { This is proved  by induction. It is true for $n=2$ as $\hat{F}_{2, +}=\hat{X}$ and $\hat{F}_{2, -}=\hat{Y}$ and $\hat\delta_2^+=1+\bar a\bar b\hat{Y}_A\hat{Y}_B$, $\hat\delta_2^-=\bar a\hat{Y}_A+\bar b\hat{Y}_B$ while $\epsilon_2=ab$. Suppose it is true for $n-1$ and for $n$ we have
\begin{eqnarray*}
\hat{F}_n^2&=&\hat{F}_{n-1,+}^2+\hat{F}_{n-1,-}^2+\hat{F}_{n-1,+}\hat{F}_{n-1,-}\hat{A}_n\hat{A}'_n+\hat{F}_{n-1,-}\hat{F}_{n-1,+}\hat{A}'_n\hat{A}_n\\
&=&4\hat\delta_{n-1}^++(2i\delta^{-}_{n-1} )(\hat{A}'_n\hat{A}_n-\hat{A}_n\hat{A}'_n)\\
&=&4\hat\delta_{n-1}^++4\hat\delta^-_{n-1}\bar a_n\hat{Y}_n\\
&=&4\hat\delta_n^+
\end{eqnarray*}
where  we have used 
$\hat{F}_{n-1,+}\hat{F}_{n-1,-}=\frac{1}{4}(\hat{F}^{2}_{n-1}-\hat{F}'^{2}_{n-1}+\hat{F}'_{n-1}\hat{F}_{n-1}-\hat{F}_{n-1}\hat{F}'_{n-1}) =-2i\delta^{-}_{n-1}$ and 
$\hat{F}_{n-1,-}\hat{F}_{n-1,+}=2i\delta^{-}_{n-1}$,  $\hat{A}_{n-1}\hat{A}_{n-1}'=a_{n-1}+i\bar a_{n-1}\hat{Y}_{n-1}$ and $\hat{A}'_{n-1}\hat{A}_{n-1}=a_{n-1}-i\bar a_{n-1}\hat{Y}_{n-1}$, and $\hat{\delta}^{+}_{n}=\frac{1}{2}(\hat{\delta}^{+}_{n-1}+\hat{\delta}^{-}_{n-1})(1+\bar a_{n} \hat{Y}_{n})+\frac{1}{2}(\hat{\delta}^{+}_{n-1}-\hat{\delta}^{-}_{n-1})(1-\bar a_{n} \hat{Y}_{n})=\hat{\delta}^{+}_{n-1}+\hat{\delta}^-_{n-1}\bar a_{n} \hat{Y}_{n}$. }
\item {\rm Proof $\hat{F}_n\hat{F}_n'=4(\epsilon_n+i\hat \delta_n^-)$} 
By the definition  {$\hat{F}_n=\hat{F}_{n-1}\frac{\hat{A}_n+\hat{A'}_n}2+\hat{F}_{n-1}'\frac{\hat{A}_n-\hat{A'}_n}2$,  and  $  \hat{F}_n'=\hat{F}_{n-1}'\frac{\hat{A}_n+\hat{A}'_n}2-\hat{F}_{n-1}\frac{\hat{A}_n-\hat{A}'_n}2,$ we have 
\begin{eqnarray*}
\hat{F}_n\hat{F}_n'=&&\frac{1}{4}[\hat{F}^{'2}_{n-1}({\hat{A}_n-\hat{A'}_n})({\hat{A}_n+\hat{A'}_n})-\hat{F}^{2}_{n-1}({\hat{A}_n+\hat{A'}_n})({\hat{A}_n-\hat{A'}_n})\\
&&-\hat{F}_{n-1}'\hat{F}_{n-1}(\hat{A}_n-\hat{A'}_n)(\hat{A}_n-\hat{A'}_n)+\hat{F}_{n-1}\hat{F}'_{n-1}(\hat{A}_n+\hat{A'}_n)(\hat{A}_n+\hat{A'}_n)]\\
=&&\frac{1}{4}[2\hat{F}^{'2}_{n-1}[{\hat{A}_n, \hat{A'}_n}]-\hat{F}_{n-1}'\hat{F}_{n-1}(2\openone-\hat{A}_{n}\hat{A}'_{n}-\hat{A}'_{n}\hat{A}_{n})\\
&&+\hat{F}_{n-1}\hat{F}'_{n-1}(2+\hat{A}_n\hat{A'}_n+\hat{A}'_n\hat{A'}_n)]\\
=&&4\hat{\delta}^{+}_{n-1}(ia_{n}Y)+(\epsilon_{n-1}+i\hat \delta_{n-1}^-)(2\openone+a_{n})-(\epsilon_{n-1}-i\hat \delta_{n-1}^-)(2\openone-a_{n})\\
=&&4\hat{\delta}^{+}_{n-1}(ia_{n}Y)+4\epsilon_{n}+4i\delta_{n-1}^-\\
=&& 4(\epsilon_n+i\hat \delta_n^-)
\end{eqnarray*}
where  we have used  $\hat{\delta}^{-}_{n}=\frac{1}{2}(\delta^{+}_{n-1}+\delta^{-}_{n-1})(1+\bar a_{n} \hat{Y}_{n})-\frac{1}{2}(\delta^{+}_{n-1}-\delta^{-}_{n-1})(1-\bar a_{n} \hat{Y}_{n})=\hat{\delta}^{-}_{n-1}+\hat{\delta}^{+}_{n-1}\bar{a}_{n}\hat{Y}_{n}$.} 

\item The proof of  $\hat{F}_{n,\pm}^2=2(\hat\delta_n^+\pm \epsilon_n)$ is straightforward.
\end{itemize}
 
\end{proof}

\subsubsection{Main results }
In a general state, it holds
\begin{eqnarray}
\langle \hat{F}_n\cos t+\hat{F}_n'\sin t\rangle^2&\le& \langle (\hat {F}_n\cos t+\hat {F}_n'\sin t)^2\rangle \nonumber\\
&=& 4(\langle\hat\delta_n^+\rangle+\epsilon_n\sin2t)\nonumber \\ 
&\le &4(\delta_n+\epsilon_n\sin 2t) \label {gel}
\end{eqnarray}
for any $t$, where $$\delta_n=\frac12\prod_{j=1}^n(1+\bar a_j)+\frac12\prod_{j=1}^n(1-\bar a_j).$$
and  {we have used $\hat{F}_n^2=\hat{F}_n'^2=4\hat \delta_n^+,\quad \hat{F}_n\hat{F}_n'=4(\epsilon_n+i\hat \delta_n^-),$  and $\hat{F}'_n\hat{F}_n=4(\epsilon_n-i\hat \delta_n^-)$.}

ii) as $\{\hat{F}_{n, +},\hat{F}_{n, -}\}=0$,  we have
$$\frac{\langle \hat{F}_{n,+}\rangle^2}{\delta_n+\epsilon_n}+\frac{\langle \hat{F}_{n,-}\rangle^2}{\delta_n-\epsilon_n}\le 2,$$
so that there exist $0\le r\le 1$ and $u$ such that
$$\langle \hat{F}_{n,+}\rangle=r\sqrt{2(\delta_n+\epsilon_n)}\cos u,\quad \langle \hat{F}_{n, -}\rangle=r\sqrt{2(\delta_n-\epsilon_n)}\sin u$$

\textbf{Symmetry.} For any subset of $k$ qubits, we have
$$\hat{F}_n=\frac12\hat{F}_k \hat{F}_{k',+}+\frac12 \hat{F}_k' \hat{F}_{k',-},\quad \hat{F}_n'=\frac12\hat{F}_k'\hat{F}_{k',+}-\frac12 \hat{F}_k \hat{F}_{k',-}$$
Thus, for any possible cut $\varrho=\rho_k\otimes\rho_{k'}$ 
\begin{eqnarray}
\langle\hat{F}_n\cos t+\hat{F}_n'\sin t\rangle_\varrho&=&\langle \hat{F}_{k', +}\rangle_{k'}\frac{\langle \hat{F}_{k}\cos t+\hat{F}_{k}'\sin t\rangle_{k'}}2+\langle \hat{F}_{k', -}\rangle_{k'}\frac{\langle \hat{F}_{k}'\cos t-\hat{F}_{k}\sin t\rangle_{k'}}2  \nonumber \\
&\le&\sqrt{\frac{\delta_{k'}+\epsilon_{k'}}2F_k(t)^2+\frac{\delta_{k'}-\epsilon_{k'}}2F_k'(t)^2} \nonumber \\
&=&\sqrt{\delta_{k'} \frac{F_k^2+F_k'^2}2+\epsilon_{k'} \frac{(F_k^2-F_k'^2)\cos2t+2F_kF_k'\sin 2t}2}  \nonumber  \\
&=&\sqrt{\delta_{k'} (F_{k,+}^2+F_{k,-}^2)+\epsilon_{k'} (2F_{k,+}F_{k,-}\cos2t+(F_{k,+}^2-F_{k,-}^2)\sin 2t)} \nonumber  \\
&=&\sqrt{2\delta_{k'} (\delta_k,+\epsilon_k\cos 2u)+2\epsilon_{k'} (\sqrt{\delta_k^2-\epsilon_k^2}\sin2u\cos2t+(\delta_k\cos 2u+\epsilon_k)\sin 2t)}  \nonumber   \\
&\le&\sqrt{2(\delta_{k'} \delta_k,+\epsilon_{k'}\epsilon_k\sin 2t)+2\sqrt{(\delta_k\epsilon_{k'}\sin2t+\delta_{k'}\epsilon_k)^2+\epsilon_{k'}^2(\delta_k^2-\epsilon_k^2)\cos^22t}}  \nonumber \\
&=&\sqrt{2(\delta_{k'} \delta_k,+\epsilon_{k'}\epsilon_k\sin 2t)+2\sqrt{(\delta_{k'} \delta_k,+\epsilon_{k'}\epsilon_k\sin 2t)^2-(\delta_k^2-\epsilon_k^2)(\delta_{k'}^2-\epsilon_{k'}^2)}} \label{kkkt}
\end{eqnarray}

\subsubsection {  Qutrit-case:  Bounds  of  $F_{3}$ and $F_{3}+F'_{3}$  }

\emph{A general state:}
We first consider $F_{3}$ let $t=0$  using  Eq.{\ref{gel}}.  For a tri-partite system, we let  $n=3$, $k=3$. We have the relation: 
$$F_{3}\leq \sqrt{2\delta_{ 3} }  =2\sqrt{1+\bar a\bar b+\bar a\bar c+\bar b\bar c}.$$
Consider the $F_{3}+F'_{3}$, we have let $t=\frac{\pi}{4}$, we have the upper bound as 
$$F_{3}+F'_{3}\leq 2\sqrt{\delta_3+\epsilon_3}=2\sqrt{2(1+\bar a\bar b+\bar a\bar c+\bar b\bar c+abc)}.$$
\emph{Partially entangled state:} In this case, we let  $k=2$, and $k'=1$ and use Eq.{\ref{kkkt}}
\begin{eqnarray}
\langle\hat{F}_3\rangle_{(21)} &\le &\max_{(a,b,c)} \sqrt{2\delta_{1} \delta_2+2\sqrt{\delta^{2}_{1} \delta^{2}_2-(\delta_2^2-\epsilon_2^2)(\delta_{1}^2-\epsilon_{1}^2)}} \nonumber \\
&=&\max_{(a,b,c)}\sqrt{2(1+\bar a\bar b)+2\sqrt{(1+\bar a\bar b)^2-(\bar a+\bar b)^2\bar c^2}}({\ge 2})
\end{eqnarray}
and the  maximum bound for bipartition separable state  $C|AB$ is 
\begin{eqnarray}
\langle F_{(3)}+F'_{(3)}\rangle_{(21)}=\sqrt{2(\delta_{1}+\epsilon_{1})(\delta_{2}+\epsilon_{2}) }=2\sqrt{(1\pm c)(1+\bar a\bar b\pm ab)}. 
\end{eqnarray}
\emph{Fully separable state:}
 We consider $t=\frac \pi4$ and use the bounds of  $F_{2}$, $F'_{2}$ dealing with separable states in Eq.(\ref{kkkt}),   which leads to 
$$\langle \hat{F}_3+\hat{F}_3'\rangle_{(111)}\le \sqrt 2\sqrt{1+abc+(c+ab)z+\sqrt{(1+abc+(c+ab)z)^2-\bar a^2\bar b^2(1+cz)^2}}:=\sqrt{2f(z)}.$$
For convenience, we denote $u=a^2+b^2+c^2$, $v=a^2b^2+a^2c^2+b^2c^2$, and $w=abc$ and
$$R=\{(a,b,c)\mid a^2b^2+a^2c^2+b^2c^2-a^2b^2c^2+2abc\sqrt{a^2+b^2+c^2+2abc}\le 0\}.$$
\begin{figure}[ht]
    \centering
    \includegraphics[height=9cm]{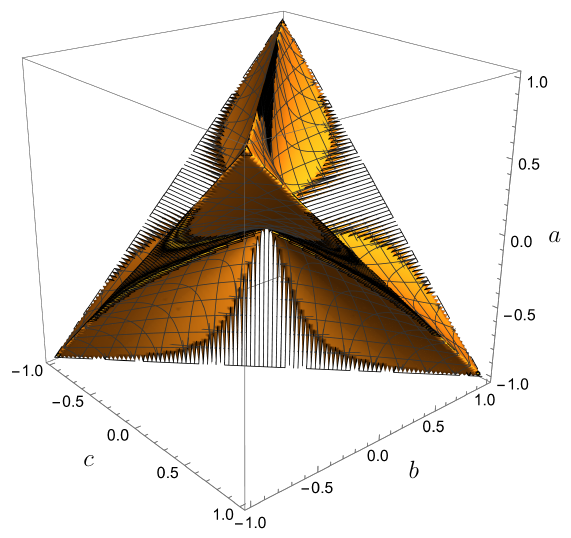}
    \caption{Depending on whether the parameters \(a, b, c\) belong to the set $R=\bigl\{(a,b,c) \mid a^2b^2+a^2c^2+b^2c^2-a^2b^2c^2+2abc\sqrt{a^2+b^2+c^2+2abc}\le 0\bigr\},$
the upper bound for \(\langle\hat{F}_3+\hat{F}_3'\rangle_{(111)}\) takes different expressions.}
\end{figure}

Then the upper bound reads
$$\langle\hat{F}_3+\hat{F}_3'\rangle_{(111)} \le \left\{\begin{array}{ll}\displaystyle\max_{\mu,\nu=\pm}\sqrt{2(1+\mu a)(1+\nu b)(1+\mu\nu c)}&(a,b,c)\not\in R \\ 
\displaystyle2\bar a\bar b\bar c\left(2+\frac{ab}c+\frac{ac}b+\frac{bc}a-abc\right)^{-\frac12}&(a,b,c)\in R\end{array}\right.$$
\begin{itemize}
\item Observation: $a^2b^2+a^2c^2+b^2c^2-a^2b^2c^2=a^2+b^2+c^2-1+\bar a^2\bar b^2\bar c^2\ge u-1$. If $(a,b,c)\in R$ then $w\le0$
$$(\sqrt{u+2w}-{|w|})^2-(1-{|w|})^2\le 0$$
which infers $\sqrt{u+2w}\le 1$ yielding $v-w^2+2w\le0$ in the case of $w<0$ which means $2-w+v/w\ge0$.
\item Proof. We denote
$$ f(z)=L+Kz+\sqrt{(L+Kz)^2-p^2(1+cz)^2}=L+Kz+\sqrt{\alpha+2\beta z+\gamma z^2}$$
with $L=1+abc, K=c+ab$, and $p=\bar a\bar b$ and $\alpha=L^2-p^2,\beta=KL-p^2c$, and $\gamma=K^2-p^2c^2$. As $f''(z)\propto \alpha\gamma-\beta^2=-p^2(k,-c L)^2\le 0$ we have that global maximum is attained at $$0=f'(z)=k,+\frac{\beta+\gamma z}{\sqrt{\alpha+2\beta z+\gamma z^2}}\Longrightarrow z_*=\frac{-\beta-K(1-c^2)\frac{|ab|}{|c|}}\gamma.$$
We see that $|z_*|\le 1$ iff $(a,b,c)\in R$ in which it holds $abc<0$. Otherwise the maximum is attained at the boundary $z=\pm1$.
\end{itemize}
In the case of $F_3$ we have
\begin{eqnarray*}
\langle \hat{F}_3\rangle_{(111)}&\le& \max_u\sum_\pm\sqrt{\frac{(1\pm\bar a\bar b)(1+c\cos u)+ab\bar c\sin u}2}
\end{eqnarray*}

The trade-offs between these upper bounds for   $\hat{F}_{3}+\hat{F}'_{3}$, defined as  $f_{(21)}(v) \equiv \max_{\{\Omega | {\rm U}(\Omega) = v\}} {\rm U}_{(21)}(\Omega)$ and  $f_{(111)}(t) \equiv \max_{\{\Omega | {\rm U}(\Omega) = v\}} {\rm U}_{(111)}(\Omega)$
for rank one projective measurements,  are shown in the following figture.  To extend these trade-offs to general measurements, we  observe that all three bounds remain invariant under sign changes of parameters $a, b, c$ when $|a|=|b|=|c|=1$. Through the convexity and monotonicity properties of $f_{(21)}$ and $f_{(111)}$ established previously,  {it follows from  Eq.(1) in maintext  that } these trade-off relations hold for arbitrary measurement schemes.
\begin{figure}
\begin{center}
\includegraphics[width=0.400\textwidth]{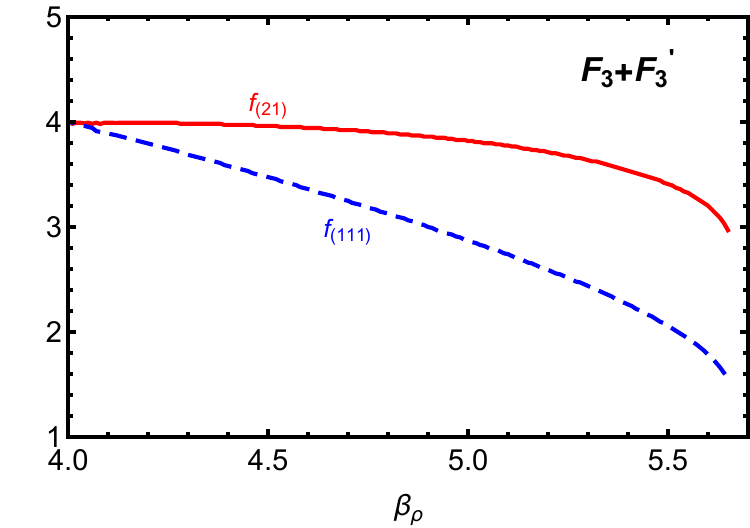}
\end{center}
\caption{Structure functions for Bell's inequality ${F}_{3}+F'_{3}$  for bipartition separable state and fully separable state. With devices that can generate a Bell value $\beta_{\rho}$, the strengthened  bound for bipartite separable  states is given as $f_{(21)}$ (red solid curve) and the bound for fully separable states is given  as $f_{(111)}$ (blue dashed curve). Any correlation exceeding these bounds can certify tripartite genuine entanglement and not fully separable states.}\label{fig1}
\end{figure}

\subsection{A.5 Bounds of general Bell's inequalities for qubit systems } 
 {In the main text, we have mainly focused on Bell inequalities that deal with full correlations, and we have derived analytic expressions for the trade-off between the bounds for separable states and those for general quantum states. For more general Bell inequalities, which may contain marginal terms and do not assume such a high symmetry, analytic results are not easy to obtain. This does not necessarily imply, however, that Bell inequalities cannot be tightened for separable states. We provide a numerical approach to explore this aspect. The key ingredient remains to derive upper bounds for entangled states and separable states using projective measurements, and then to explore the trade-off between them. If the relation ${\rm U}_{\varrho}(\Omega^{*}_{i}) \leq f(\max_{\pm} {\rm U}(\Omega_{i, \{\pm\}}))$ holds, then one can tighten the bounds for separable states. }

 {First, we consider measurements that are rank-one projective, as specified by the set $\Omega$. We then calculate the upper bounds for an entangled state and a separable  state $\varrho$, denoted by $\rm U$ and ${\rm U}_{\varrho}$ respectively.
Our main tool remains uncertainty relation Eq.(\ref{UR}), which can be re-expressed as  
$$A^{2}_{0}+A^{2}_{1}+a^{2}-2aA_{0}A_{1}\leq 1.$$
 Specifically, consider a general Bell scenario. The experimental statistics are described by the set}
\[
\mathbf{p}_{\mathbf{s}} \equiv \left\{ p(\mathbf{a}|\mathbf{s}) = {\rm Tr}\left( \rho \bigotimes_{l=1}^n M^{a'_l}_{s_l} \right) \right\}_{\mathbf{a}},
\]
 {where $M^{a_m}_{s_m}$ denotes the measurement operator for outcome $a_m$ given setting $s_m$ on the $m$-th subsystem. Here, $\mathbf{s} \equiv (s_1, \dots, s_n)$ represents the settings and $\mathbf{a} \equiv (a'_1, \dots, a'_n)$ the outcomes. One can interpret the experimental statistics as follows. Conditioning on the setting $\mathbf{s}_{\bar{m}} \equiv (\dots, s_l, \dots)_{l \neq m}$ and the outcomes $\mathbf{a}_{\bar{m}} \equiv (\dots, a'_l, \dots)_{l \neq m}$ of the other $n-1$ parties, a state}
\[
\rho_{\mathbf{a}_{\bar{m}}|\mathbf{s}_{\bar{m}}} = \frac{ {\rm Tr} \left[ \rho \bigotimes_{l \neq m} M^{a'_l}_{s_l} \right] }{ p(\mathbf{a}_{\bar{m}}|\mathbf{s}_{\bar{m}}) }
\]
is prepared on the $m$-th party with probability $p(\mathbf{a}_{\bar{m}}|\mathbf{s}_{\bar{m}}) = \sum_{a'_m} p(a'_m, \mathbf{a}_{\bar{m}} \mid s_m, \mathbf{s}_{\bar{m}}).$ On this local state, performing a binary measurement $\hat{A}_{s_m}$ yields the expectation value
\[
A_{s_m | \rho_{\mathbf{a}_{\bar{m}}|\mathbf{s}_{\bar{m}}}} = \frac{ p(0, \mathbf{a}_{\bar{m}} \mid s_m, \mathbf{s}_{\bar{m}}) - p(1, \mathbf{a}_{\bar{m}} \mid s_m, \mathbf{s}_{\bar{m}}) }{ p(\mathbf{a}_{\bar{m}}|\mathbf{s}_{\bar{m}}) }.
\]
 {Thus, the correlation then is linked to the local statistics of measurements \(\hat{A}_{1}\) and \(\hat{A}_{0}\) on the conditional states \(\rho_{\mathbf{a}_{\bar{m}}|\mathbf{s}_{\bar{m}}}\), all of which are governed by uncertainty relations. }

 {Let us consider the bound for a separable state, which can be assumed to take the form $\varrho=\rho_{k} \otimes \rho_{k'}$. For such a product state, the probability distribution factorizes as $p(\mathbf{a}|\mathbf{s}) = p(\mathbf{a}_{(k)}|\mathbf{s}_{(k)}) p(\mathbf{a}_{(k')}|\mathbf{s}_{(k')})$. By applying the above constraints to $p(\mathbf{a}_{(k)}|\mathbf{s}_{(k)})$ and $p(\mathbf{a}_{(k')}|\mathbf{s}_{(k')})$, one can use this approach to calculate the upper bounds for a general state $\rho$ and the separable state $\varrho$, denoted by ${\rm U}(\Omega)$ and ${\rm U}_{\varrho}(\Omega)$ respectively.}

 {Second,    the  trade-off between the bounds for separable state and general state $f$ for projection measurement can also be calculated as 
$$f_{\varrho}(v)=\max_{\{\Omega|{\rm U}(\Omega)=v\}}{\rm U}_{\varrho}(\Omega) $$ when  $f$ is concave  and monotonically decreasing and holds ${\rm U}_{\varrho}(\Omega^{*}_{i}) \leq f(\max_{\pm} {\rm U}(\Omega_{i, \{\pm\}}))$,  the trade-off stands also for general measurement settings.  One thus has  ${\rm U}_{\varrho}(\Omega')\leq  f({\rm U}(\Omega'))$ by the argument in main text. }

\section{B.  NPA hierarchy}

 { The NPA hierarchy is a sequence of tests to verify whether a correlation can be generated by a quantum system. It considers a set $\mathcal{O}=\{O_{e}\}$ of operators, where each $O_{e}$ specifies a measurement operator, or a combination or product of such operators. Consider an arbitrary linear combination $W=\sum_{e}\alpha_{e}O_{e}$ with coefficients $\mathbf{C}=(\alpha_{1}, \cdots, \alpha_{n'})^T$. It follows that $\mathrm{Tr}(\rho W^{\dagger} W)=\mathbf{C}^\dagger\Gamma \mathbf{C}\geq 0$ for any $\mathbf{C}$, where the matrix elements are given by $\Gamma_{ef}= \mathrm{Tr}(\rho O^{\dagger}_{e}O_{f})$. As a result, the matrix $\Gamma$ must be positive semidefinite.}

 {Note that some entries of $\Gamma$ correspond to experimental statistics. For instance, if $O_{e}=A_{e}$ and $O_{f}=B_{f}$ are measurement operators for two different parties, then $\Gamma_{ef}=\langle A_{e} B_{f} \rangle$ can be measured in experiments. Some matrix elements $\Gamma_{ef}$ correspond to non-commuting operators. For example, for $A_{e}$ and $A_{f}$, one has $\Gamma_{ef}=\mathrm{Tr}(\rho A^{\dagger}_{e}A_{f})$, which does not correspond to any experimental statistic and is called a non-physical term. If a correlation arises from a quantum experiment, one can populate the physical terms in $\Gamma$ with the measured values and assign values to the non-physical terms such that the matrix is  semidefinite. Otherwise, if for a given set of physical terms placed in $\Gamma$, no assignment of values to the non-physical terms can make the matrix positive semidefinite, then the correlation cannot be realized by a quantum system.}  
 
 {To further  verify whether a correlation can be generated with separable state, one can use structure specific to separable states. 
 First, when state is product $\varrho=\rho^{(A)} \otimes \rho^{(B)}$, one can   define $\mathcal{O}$ as a set of $\{O^{(A)}_{e} \otimes O^{(B)}_{f}\}$, with $O^{(A)}$ is  relevant to Alice's system, and $O^{(B)}_{f}$ for Bob's side. Then, one can define a  matrix $\Gamma$ with entries  
 \begin{eqnarray*}
\Gamma_{ee', ff'} &=& {\rm Tr}[\varrho (O^{(A)\dagger}_e O^{(A)}_{e'} + O^{(A)\dagger}_{e'} O^{(A)}_e)\otimes(O^{(B)\dagger}_f O^{(B)}_{f'} + O^{(B)\dagger}_{f'} O^{(B)}_f)]\\
&=&  {\rm Tr}[\rho^{(A)} (O^{(A)\dagger}_e O^{(A)}_{e'} + O^{(A)\dagger}_{e'} O^{(A)}_e)]\otimes{\rm Tr}[\rho^{(B)}(O^{(B)\dagger}_f O^{(B)}_{f'} + O^{(B)\dagger}_{f'} O^{(B)}_f)]\nonumber \\
&=&\Gamma^{(A)}_{e'e}\Gamma^{(B)}_{ff},
\end{eqnarray*}
where $ee'$ and $ff'$ specify the index of lines and columns of this entry in the matrix.  Then $\Gamma=\Gamma^{(A)}\otimes \Gamma^{(B)} \ge 0$ as $\Gamma^{(A), (B)}$ are semi-definite. However, the decomposition in the second step does not apply to entangled state, and $\Gamma$ while can be well-defined may be not semi-definite.  
Thus,   if correlation can be generated with separable  state, one put the physical terms in the matrix and  assign some number to  the non-physical terms to ensure that this matrix are semi-definite.  Other wise, it cannot be generated using separable state.} 

  \emph{An example:}

 {To demonstrate that the above approaches enables the verification of entanglement using local correlations, we provide an example in the minimal Bell scenario. Here, Alice and Bob perform orthogonal, rank one projective measurements specified by the parameter set $\Omega = \{0, 0\}$. We consider a correlation of the form:} \[
p(a'b'\mid \nu \mu) = \frac{1 + \lambda (-1)^{a'+b'+ \nu \mu}}{4}.
\]
 {(This correlation  can be achieved with state $\rho_{\lambda}=\lambda |\Phi\rangle \langle \Phi|+(1-\lambda )\openone/4 $  with  $|\Phi\rangle =\sqrt{2}/2(|00\rangle +|11\rangle )$ using measurements $(\sigma_{x}\pm \sigma_{z})/\sqrt{2}$.)  The corresponding Bell value for the CHSH inequality is \(4\lambda\).   The correlation becomes nonlocal when \(\lambda \geq \frac{\sqrt{2}}{2}\). By applying the method described in this section and substituting the relevant parameters into the modified NPA matrix---specifically, \(r^{*}_{\nu} = r^{*}_{\mu} = 0\), \(r_{\nu} = r_{\mu} = 1\), and \(a = b = 0\)---we find that when  \(\lambda \geq 0.447\) or $\beta_{\rho_{\lambda}}\geq 1.788$,  one can not assign  some  numbers to the non-physical terms in $\Gamma$ as defined above  such that the matrix is semi-definite, and entanglement thus can be  certified. }

\end{document}